\documentclass[11pt,a4paper]{article}
\usepackage[utf8]{inputenc}
\usepackage[T1]{fontenc}
\usepackage[english]{babel}
\usepackage{lmodern,microtype}
\usepackage{amsmath,amssymb,amsthm,mathtools}
\usepackage[margin=2.5cm]{geometry}
\usepackage{booktabs,array,longtable,graphicx}
\usepackage{enumitem,needspace}
\usepackage[hidelinks]{hyperref}
\setlist{itemsep=1pt,topsep=3pt,leftmargin=*}
\newtheorem{theorem}{Theorem}[section]
\newtheorem{proposition}[theorem]{Proposition}
\newtheorem{lemma}[theorem]{Lemma}
\newtheorem{corollary}[theorem]{Corollary}
\theoremstyle{definition}
\newtheorem{definition}[theorem]{Definition}

\newcommand{\hash}{\mathsf{H}}
\newcommand{\Enc}{\mathsf{Enc}}

\title{RecHash: Segmented Hash-Entangled Chains\\for Stateful Sequential Authentication}
\author{Luis Adrián Lizama-Pérez\\
\small Universidad Autónoma Metropolitana, Lerma Campus, Mexico\\
\small Department of Information Systems and Communications\\
\small\texttt{l.lizama@ler.uam.mx}}
\date{Version 3.0 --- October 29, 2026}
\begin{document}\maketitle
\begin{abstract}
RecHash is a stateful sequential authentication construction based on paired hash-chain openings and an authenticated record history. It develops the decreasing position relation of the 2019 hash-entangled construction, whose exact integer normalization yields a two-coordinate constant-sum antichain. Direct proportional reuse is constrained by integer capacity and multiplicative contraction. The main construction instead assigns each record to a predetermined chain segment, with complementary positions encoding each digest block. All digests fit every scheduled slot, independently of preceding messages. For an $\ell$-bit digest distributed over $k$ equal-length chain pairs, balanced allocation gives $\lfloor L/(2^{\lceil\ell/k\rceil}+1)\rfloor$ slots. A 256-bit digest, 32 pairs, and length 65535 provide 254 application records followed by a renewal certificate, with 2048 bytes of openings per record for 256-bit chain outputs. We give a complete segmented state transition, a worked trace, exact allocation and resource results, and a renewal-free security-composition theorem conditional on authenticated history and an adaptive one-step guarantee. The analysis relates the construction to constant-sum encodings and standardized hash-based signatures; a concrete post-quantum reduction and security across renewal remain open.
\end{abstract}
\section{Introduction}
RecHash studies sequential authentication using pairs of hash chains whose disclosed values become the verification heads for subsequent records. Its principal construction assigns successive messages to disjoint segments of the chains. Within each segment, a digest block selects two complementary positions: increasing one position decreases the other. Across segments, every new position lies before the positions disclosed in the preceding segment. This arrangement preserves the local constraint against forward substitution while providing a predetermined number of records per chain-generation phase.

The construction addresses a particular setting: one signer emits an ordered stream of records, and verifiers obtain its authoritative current state from an authenticated-history service, such as Singlechain or an equivalent mechanism. The service fixes order and membership; the chain openings authenticate a candidate extension. The signer waits for acceptance before disclosing the next record. This setting makes state synchronization part of the protocol and distinguishes the resulting records from signatures intended for independent verification under a fixed public key.

The historical starting point is the hash-entangled construction published in 2019~\cite{Lizama2019}, which combined digest partitioning, a decreasing relation between chain positions, and Singlechain history. Its HMAC extension appeared online in 2021 and in a 2022 journal volume~\cite{Lizama2022}; the 2023 review makes the relevant opening convention explicit~\cite{Lizama2023}. We retain this chronology and recover the exact integer relation behind the construction. Under a positive coordinate normalization, its admissible position pairs form a constant-sum antichain. Constant-sum and incomparable encodings also appear in earlier and later signature research~\cite{BosChaum1993,Zhang2023,Tiwari2026}. The equivalence established here concerns that local algebraic structure, not the complete protocols or their security models.

The paper develops two allocation profiles with different roles. DIRECT follows proportional recursion and exposes the capacity restrictions caused by integer positions and repeated contraction. SEGMENTED is the main operational construction: it schedules translated constant-sum codes in advance, so every digest fits each designated slot regardless of earlier messages. DIRECT supplies the analytical comparison explaining this design choice. No automatic transition between the profiles is part of the protocol.

For an $\ell$-bit digest divided among $k$ chain pairs of equal length $L$, balanced partitioning gives $\lfloor L/(2^{\lceil\ell/k\rceil}+1)\rfloor$ slots. With $\ell=256$, $k=32$, and $L=65535$, the construction uses 64 physical chains and provides 255 slots: 254 application records and one renewal certificate under the continuing-operation policy. Each record contains 64 hash openings, or 2048 bytes when chain outputs are 256 bits, in addition to its context and payload.

The results comprise the historical algebraic normalization, exact capacity and allocation bounds, a state-transition specification, and analytic resource counts. Security is expressed through a conditional composition theorem: an adaptive one-step bound, together with authenticated history and exclusive state use, bounds accumulated impersonation probability over an admissible renewal-free horizon. Establishing that one-step bound from concrete hash assumptions, including quantum adversaries, remains a separate task. The following sections place the construction in context, define and illustrate the segmented protocol, analyze proportional allocation, and develop the security and parameter results.
\section{Historical development and related constructions}\label{sec:related}
The relation used in the 2019 construction can be recovered by separating opening positions from verification distances. Let $L$ be the signing-chain length, $Q=L-h_0$ the previous verifier anchor, and $p=L-h_1$ the signing position. The quantity $\delta=h_1-h_0h_1/L$ in that construction satisfies
\begin{equation}\label{eq:historical-weighted}
q=\delta=Q(1-p/L),\qquad Qp+Lq=LQ.
\end{equation}
Some original instructions publish $f^{Q-\delta}(s_0)$ and traverse by $\delta$, interchanging the opening index and its complementary distance. The convention realizing~\eqref{eq:historical-weighted} publishes $f^\delta(s_0)$ and traverses by $Q-\delta$; it is explicit in the 2023 review~\cite{Lizama2019,Lizama2023}. The HMAC development adds reciprocal authentication and a historical transaction database~\cite{Lizama2022}. The normalization below applies to the recovered decreasing relation, with integer admissibility stated explicitly.

Coordinatewise order means $(p,q)\preceq(p',q')$ if $p\le p'$ and $q\le q'$. An antichain contains no two distinct comparable elements. Here this order describes forward traversal of indices; collisions or cycles in concrete hash functions are separate computational events.

\begin{theorem}[Exact lattice normalization]\label{thm:equivalence}
Let $L,Q$ be positive integers, $d=\gcd(L,Q)$, and define
\[
\mathcal E_{L,Q}=\{(p,q)\in\mathbb Z_{\ge0}^2:Qp+Lq=LQ\},\qquad
\mathcal C_d=\{(u,v)\in\mathbb Z_{\ge0}^2:u+v=d\}.
\]
Then
\begin{equation}\label{eq:isomorphism}
\Phi(p,q)=\left(\frac dL p,\frac dQ q\right)
\end{equation}
is a bijection from $\mathcal E_{L,Q}$ to $\mathcal C_d$. It preserves and reflects coordinatewise order. Both sets have $d+1$ elements, or $d-1$ when both coordinates are required to be positive. In normalized coordinates, the recovered 2019 relation is exactly a two-coordinate constant-sum antichain.
\end{theorem}
\begin{proof}
Write $L=d\ell$ and $Q=dr$ with $\gcd(\ell,r)=1$. The defining equation becomes $rp+\ell q=d\ell r$. Reduction modulo $\ell$ gives $\ell\mid p$, so $p=\ell a$ and $q=r(d-a)$ for $a=0,\ldots,d$. Applying $\Phi$ gives $(a,d-a)$, and every such pair has the inverse image $(\ell a,r(d-a))$. Positive coordinate scalings preserve and reflect each coordinate inequality, both on this set and on the corresponding rectangular sublattice. If two distinct elements of $\mathcal C_d$ were comparable, their nonnegative coordinate differences would sum to zero, a contradiction. The endpoint and interior counts follow from the allowed values of $a$.
\end{proof}

When $Q=L$, the normalization is the identity and the relation is $p+q=L$. This establishes a precise temporal and mathematical connection. Decreasing coupling and sequential history occur in the 2019 development; the later review clarifies the opening convention; the present formulation identifies the associated integer lattice and its constant-sum representation. Bos and Chaum's fixed-cardinality incomparable encodings predate that development~\cite{BosChaum1993}. Constant-sum Winternitz encodings were studied by Zhang, Cui, and Yu in 2023~\cite{Zhang2023}. The 2026 Antichain Winternitz (ACW) preprint describes constant-sum chain witnesses with label-recovery functionality for Bitcoin~\cite{Tiwari2026}. Our characterization of ACW is limited to its official abstract and metadata. A pairwise antichain is a local mechanism shared across these settings; equality of complete encodings, protocol equivalence, bibliographic influence, and transfer of security reductions do not follow from the normalization.

Lamport's one-time signature uses one-way functions to authenticate a message~\cite{Lamport1979}; Merkle's certification construction authenticates many signing instances under a compact commitment~\cite{Merkle1990}. WOTS+ develops Winternitz chain encodings with a security analysis~\cite{Hulsing2013}. Its message digits are accompanied by checksum digits, which oppose forward traversal to an alternative message~\cite{RFC8391}. RecHash instead pairs each block with a complementary position. Both designs constrain forward substitution, but their encodings and permitted reuse differ.

XMSS and XMSS$^{\rm MT}$ combine WOTS+ with tree authentication~\cite{RFC8391}; LMS combines LM-OTS with a Merkle tree, and HSS certifies a hierarchy of LMS keys~\cite{RFC8554}. NIST SP 800-208 approves specified parameter sets and prescribes stateful implementation requirements~\cite{NIST208}. Their signer state prevents reuse of one-time material. RecHash also changes the public verification state and requires authoritative chronology. An XMSS or HSS signature contains its tree authentication data; an offline verifier need not obtain an ordered log of earlier application signatures. The stateful guidance is relevant to exhaustion and rollback, but does not approve RecHash or establish its history assumption.

The SPHINCS+ framework combines FORS with a hypertree and develops a tweakable-hash abstraction for security analysis~\cite{Bernstein2019}. NIST FIPS 205, published in August 2024, standardizes the derived SLH-DSA construction with twelve SHA-2 or SHAKE parameter sets~\cite{FIPS205}. It uses FORS, WOTS+, and XMSS components while retaining a fixed public key and a stateless signing interface. RecHash's domain-separated chain functions serve a related contextual purpose, but the assumptions and reductions of SPHINCS+ or SLH-DSA do not automatically apply to them.

\begin{table}[htbp]\centering\small
\begin{tabular}{@{}p{2.35cm}p{4.25cm}p{7.1cm}@{}}\toprule
Construction & Local mechanism & Authentication across signing instances\\\midrule
WOTS+~\cite{Hulsing2013} & Chain positions and checksum & One-time component; requires an additional many-time construction.\\
XMSS / XMSS$^{\rm MT}$~\cite{RFC8391} & WOTS+ & Merkle authentication under a root; signer tracks used instances.\\
LMS / HSS~\cite{RFC8554} & LM-OTS & Tree or hierarchy of certified trees; finite signing capacity.\\
SLH-DSA~\cite{FIPS205} & FORS and WOTS+ & Stateless architecture with hypertree authentication and a fixed public key.\\
Constant-sum WOTS~\cite{Zhang2023} & Incomparable digit families & Encoding-level analysis within specified OTS frameworks.\\
RecHash & Pairwise antichains & Moving heads, ordered acceptance, and reserved renewal capacity.\\\bottomrule
\end{tabular}
\caption{Structural comparison across different abstraction levels. The table does not rank performance at equal security.}\label{tab:related}
\end{table}

The literature cutoff is October 28, 2026. The April 2026 initial public draft of NIST SP 800-230 proposes six additional SLH-DSA parameter sets for limited-use applications, with at most $2^{24}$ signatures per key~\cite{NIST230}; it remains a draft at that cutoff. It illustrates the relevance of a prescribed signing horizon, although its usage bound is not a count of physical chain segments. RFC 9708 specifies HSS/LMS use in CMS, and RFC 9909 specifies SLH-DSA conventions for X.509~\cite{RFC9708,RFC9909}. RecHash does not provide either interoperability profile.

Hash-based signatures are one branch of post-quantum authentication. FIPS 204 standardizes the module-lattice-based ML-DSA~\cite{FIPS204}, while NIST IR 8610 reports nine additional-signature candidates advancing to a third round in May 2026~\cite{NIST8610}. Those candidates are under evaluation rather than approved standards. The comparisons here focus on hash-based mechanisms and state management, not an exhaustive survey or cross-family benchmark.

A post-quantum analysis must specify adversarial access to the hash functions. Grover's quadratic search improvement does not assign a security level to an entire signature construction~\cite{Grover1996}; classical and quantum-accessible random-oracle models must be distinguished~\cite{Boneh2011}. Perlner, Kelsey, and Cooper's attack on particular pre-standardization SHA-256 SPHINCS+ parameters illustrates the importance of multi-target behavior~\cite{Perlner2022}; it is not presented here as a break of FIPS 205. For RecHash, adaptive exposures, digest binding, and the number of targets remain obligations of a profile-specific one-step reduction. The contribution developed below is the scheduling of the local antichains, its exact capacity, and its composition with sequential authenticated state.
\section{Model, notation, and cryptographic interfaces}\label{sec:model}

For each component $i\in\{1,\ldots,k\}$, family $a\in\{1,2\}$, and epoch $e$, define
\begin{equation}\label{eq:chains}
C_{a,i,e}[j]=f_{a,i,e}^{\,j}(s_{a,i,e}),\qquad 0\le j\le L_{a,i,e}.
\end{equation}
Seeds are independently sampled $\lambda$-bit strings. SEGMENTED uses equal lengths within each pair; DIRECT also permits unequal pair lengths. Indices are absolute positions from the seed within an epoch; indices in different epochs do not refer to the same physical chain.

The digest length is $\ell$ and the chain-output length is $\lambda$; both are 256 bits in the SHA-256 examples. SEGMENTED uses positive block sizes $b_i$ with $\sum_i b_i=\ell$. Its reference example has $k=32$ pairs and $b_i=8$. DIRECT uses uniform blocks, with $k=16$ and $b=16$ as its reference case. In both profiles $k$ counts pairs, so there are $2k$ physical chains and $2k$ openings per record. Small traces below explicitly use illustrative digests.

Let $\hash=\operatorname{SHA256}$ in the concrete examples and define
\[
f_{a,i,e}(x)=\hash\bigl(\Enc(\texttt{CHAIN},\mathsf{id},a,i,e,x)\bigr).
\]
The identity and registered parameters fix the profile, block allocation, maximum physical length $L_{\max}$, and payload limit $M_{\max}$. DIRECT additionally registers the positive nonce budget $N_{\max}$; SEGMENTED registers its renewal policy and uses nonce zero. These parameters remain fixed within a registration. Domain separation is specified but is not claimed to prove independence of concrete functions. Appendix~\ref{app:encoding} defines $\Enc$ without ambiguous concatenation.

The public state after round $n$ contains the immutable parameters, identity, sequence number, history commitment $T_n$, and, for every component,
\[
(e_{1,i},L_{1,i,e},Q_{n,i},A_{n,i}),\qquad
(e_{2,i},L_{2,i,e},P_{n,i},B_{n,i}),
\]
where $A_{n,i}=C_{1,i,e}[Q_{n,i}]$ and $B_{n,i}=C_{2,i,e}[P_{n,i}]$. Epoch subscripts are omitted when unambiguous. SEGMENTED also stores the number $r$ of consumed slots in the current phase; its next slot is $j=r+1$. A phase ends when all chains are renewed together. The global record count $n$ does not reset. Initially all epochs are zero and the active indices equal the respective chain lengths. Alice constructs a registration object $G$ containing parameters, identity, lengths, epochs, and heads, and sets
\[
T_0=\hash(\Enc(\texttt{GENESIS},G)).
\]
The pair $(G,T_0)$ is authenticated externally. As with an ordinary signature public key, an initial trust anchor must be established externally; RecHash subsequently authenticates changes to its verification state. The signer retains the seeds and any chain traversal cache separately from the public state.

\begin{definition}[Hash-chain opening]\label{def:opening}
For an authenticated head $H=C[t]$, an opening is the disclosed hash value
\begin{equation}\label{eq:opening}
\omega=C[z]=f^z(s),\qquad 0\le z<t,\qquad f^{t-z}(\omega)=H.
\end{equation}
The integer $z$ is its opening index, $t-z$ its verification distance, and $H$ its anchor. The opening is a value, not its index or a complete signature. RecHash uses strictly interior indices for reusable chains.
\end{definition}
For a candidate record, let
\[
\omega^{(1)}_i=C_{1,i,e}[q_i],\qquad
\omega^{(2)}_i=C_{2,i,e}[p_i],\qquad
\Omega_i=(\omega^{(1)}_i,\omega^{(2)}_i).
\]
The vector $\boldsymbol\Omega=(\Omega_1,\ldots,\Omega_k)$ supplies the $2k$ disclosed values. Acceptance of an ordinary record makes these values the next heads. A renewal record instead uses them to authorize installation of the replacement heads.

Let $X(S)=\hash(\Enc(\texttt{STATE},S))$, including the complete public state and parameters. If the prior state's counter is $n$, define
\begin{equation}\label{eq:digest}
D(S,\tau,m)=\hash\bigl(\Enc(\texttt{DIGEST},X(S),n+1,\tau,m)\bigr),
\end{equation}
where $\tau\in\{\texttt{SIGN},\texttt{RENEW}\}$ and $m$ is a byte string. Divide the digest, most-significant bit first, into the registered blocks $u_i$ of $b_i$ bits (uniform $b$ for DIRECT). The type is authenticated: the word RENEW inside an ordinary SIGN payload does not change the operation type. Components are not permuted, and ordinary messages are not modified to force admissibility.

The symbolic forward-closure model permits deriving disclosed-chain descendants by hash evaluation. Computational adversaries may also seek preimages, collisions, alternative digest encodings, and state inconsistencies. They observe all public records and choose messages adaptively. Our algebraic results concern the index model. The security composition theorem separately assumes a uniform one-step computational guarantee for reachable states and accumulated adversarial views. Neither local incomparability nor the absence of number-theoretic assumptions establishes that guarantee on its own.

\begin{definition}[Authenticated-history service, H]\label{def:H}
For each identity, the verifier has a unique, ordered, authenticated prefix of accepted records and its current state. A candidate extension must reference exactly that state and pass signature verification. Once accepted, a record cannot be replaced, reordered, or deleted without detection. Rollbacks and retrospective forks are rejected. Historical auditing verifies membership in this prefix through the service or an equivalent authenticated checkpoint.
\end{definition}
H states the required function of Singlechain or an equivalent mechanism. Equation~\eqref{eq:history} supplies a hash link but does not by itself implement the assumption. An authenticated anchor and a rule selecting the authoritative prefix are necessary.

H does not assume that every new candidate signature is legitimate. An attacker may still try to make a newly submitted record pass the cryptographic checks; that is the impersonation event analyzed below. The service preserves the accepted prefix, while the signature mechanism must justify the authenticity of new extensions. This separation avoids assuming the desired security conclusion.

\begin{table}[htbp]\centering\small
\begin{tabular}{@{}ll@{}}\toprule
Symbol & Meaning\\\midrule
$k$, $2k$ & Digest components (chain pairs), physical chains\\
$\ell$, $\lambda$ & Digest bits, bits in each chain value\\
$b_i$, $u_i$ & Block width, its integer value\\
$L_i$, $P_i,Q_i$ & Pair length, current head indices\\
$p_{j,i},q_{j,i}$ & Candidate opening indices for slot $j$\\
$s_i$, $J$ & Steps per segment, available slots in a phase\\
$n$, $r$, $e$ & Global record count, consumed slots, chain epoch\\
$\boldsymbol\Omega$ & Vector of disclosed hash values\\\bottomrule
\end{tabular}
\caption{Main notation for SEGMENTED. Indices and traversal distances are integers; openings and heads are hash values.}\label{tab:notation}
\end{table}
\section{Segmented allocation}\label{sec:segments}

In the SEGMENTED profile, $k$ counts digest components and chain pairs, so $2k$ physical chains are active. Unlike DIRECT, this profile places a translated constant-sum code inside each scheduled interval. It retains local incomparability but does not satisfy the proportional recurrence with the previous openings as endpoints. Section~\ref{sec:equivalence} analyzes that proportional alternative separately.

Let the digest length be $\ell$ bits and choose positive block sizes $b_1,\ldots,b_k$ with $\sum_i b_i=\ell$. Both chains of pair $i$ have the same registered epoch length $L_i$. Define
\begin{equation}\label{eq:segment-size}
M_i=2^{b_i},\qquad s_i=M_i+1,\qquad
J=\min_i\left\lfloor\frac{L_i}{s_i}\right\rfloor.
\end{equation}
Here $s_i$ is a number of hash steps, not a number of interior points. A segment of $s_i$ steps has exactly $M_i$ strictly interior integer positions. In epoch-local slot $j\in\{1,\ldots,J\}$, set
\begin{equation}\label{eq:segment-encoding}
\beta_{j,i}=L_i-js_i,\qquad
p_{j,i}=\beta_{j,i}+1+u_i,\qquad
q_{j,i}=\beta_{j,i}+M_i-u_i,
\quad 0\le u_i<M_i.
\end{equation}
The signing opening is at $p_{j,i}$ in chain 2 and the verifier opening is at $q_{j,i}$ in chain 1. The boundaries are numerical coordinates; their chain values are not separately published in advance. In particular, a lower boundary must not be disclosed as a segment commitment, since that disclosure would reveal all its descendants.

The local coordinates $x=p-\beta$ and $y=q-\beta$ satisfy $x+y=s_i$, with $x,y\in\{1,\ldots,M_i\}$. Equivalently, $y=s_i(1-x/s_i)$. Thus the same decreasing geometry is retained in local coordinates, whereas the global indices obey
\begin{equation}\label{eq:translated-sum}
p_{j,i}+q_{j,i}=2\beta_{j,i}+s_i.
\end{equation}
The translation and its slot number are determined by the authenticated state, not chosen by the signature submitter.

\begin{theorem}[Segment capacity and separation]\label{thm:segment-capacity}
For fixed registered lengths and block sizes, every $\ell$-bit digest has exactly one opening-index vector in each slot $1\le j\le J$. Each slot's index vectors form an antichain. Every possible opening index in slot $j+1$ is strictly below every possible index in slot $j$ in the same physical chain. Honest verification to the previously accepted heads is well defined, and all $J$ slots are available independently of the messages signed in earlier slots.
\end{theorem}
\begin{proof}
In each component, $u_i\mapsto(1+u_i,M_i-u_i)$ is injective and takes all $M_i$ inputs to interior points with fixed sum $s_i$. If two such pairs are coordinatewise comparable, their equal sums force equality in both coordinates. Applying this argument to each pair proves incomparability of distinct digest vectors. Since $j\le J$, $\beta_{j,i}\ge0$, so both indices are positive. The largest index in slot $j+1$ is $\beta_{j,i}-1$, whereas the smallest in slot $j$ is $\beta_{j,i}+1$. For the first slot the largest index is $L_i-1$, below the initial head. Hence all verification exponents to the previous heads are positive. Each accepted opening becomes the next head, and the same interval separation holds for every subsequent digest. Induction proves availability and correct traversal throughout the scheduled horizon.
\end{proof}
The slot count is exact for this specified fixed-width, strictly interior construction. It is not a universal optimum among arrangements that reuse boundaries, change encodings, or redistribute coordinates. It also does not prove computational unforgeability: equal digests and attacks on the hash functions remain outside the index argument.

\begin{theorem}[Optimal balanced allocation for equal lengths]\label{thm:balanced}
Fix $\ell$, $1\le k\le\ell$, and a common length $L$. Among positive integer block-size allocations summing to $\ell$, a balanced allocation maximizes the number of scheduled slots. Writing $\ell=ak+r$, $0\le r<k$, choose $r$ blocks of $a+1$ bits and $k-r$ blocks of $a$ bits. The resulting maximum is
\begin{equation}\label{eq:balanced}
\boxed{J(L,k,\ell)=\left\lfloor\frac{L}{2^{\lceil\ell/k\rceil}+1}\right\rfloor.}
\end{equation}
Other allocations may tie because of the floor operation.
\end{theorem}
\begin{proof}
For equal lengths, the minimum in~\eqref{eq:segment-size} is determined by $\max_i b_i$, since $\lfloor L/(2^b+1)\rfloor$ is nonincreasing in $b$. Every allocation has maximum at least $\lceil\ell/k\rceil$, by the sum constraint. The stated balanced allocation attains that value.
\end{proof}
When $k\mid\ell$, the formula reduces to $J=\lfloor L/(2^{\ell/k}+1)\rfloor$. Increasing $k$ at fixed $L$ does not decrease capacity, but the improvement is stepwise and eventually saturates: positive blocks imply $k\le\ell$, giving at most $\lfloor L/3\rfloor$ slots in this family. A plateau in $\lceil\ell/k\rceil$ gives no additional slots despite extra chains.

For a target of $R$ slots, the minimal common length in a balanced allocation is
\begin{equation}\label{eq:target-length}
L_{\min}=R\bigl(2^{\lceil\ell/k\rceil}+1\bigr).
\end{equation}
Conversely, let $h=\lfloor L/R\rfloor$. If $h<3$, no positive-bit allocation can provide $R$ slots. Otherwise put $b_{\max}=\lfloor\log_2(h-1)\rfloor$. The minimum component count is $\lceil\ell/b_{\max}\rceil$, allowing positive blocks of at most $b_{\max}$ bits. These are design formulas within the specified family.

If lengths may differ between pairs, set $L_i=R(2^{b_i}+1)$. The total chain-generation count is then $2R\sum_i(2^{b_i}+1)$. A balanced allocation minimizes this sum for fixed $k$: whenever $a\ge b+2$, replacing the sizes $(a,b)$ by $(a-1,b+1)$ decreases $2^a+2^b$ by $2^{a-1}-2^b>0$. Repeated balancing proves the claim. This optimization concerns hash counts and does not assign equal economic costs to memory, bandwidth, and hashing.

For a small trace, consider one pair, a two-bit illustrative digest, $L=15$, and $s=5$. There are three slots, with boundaries $(10,15)$, $(5,10)$, and $(0,5)$, used in that order. Suppose the two application records and the renewal payload produce block values $1,3,2$. These are illustrative digest outputs, not a rule allowing the signer to replace the digest of a chosen message. Table~\ref{tab:trace} shows every transition.

\begin{table}[htbp]\centering\small
\begin{tabular}{@{}clrrrrr@{}}\toprule
Slot & Record & $u$ & $\beta$ & $(q,p)$ & Prior $(Q,P)$ & Distances $(Q-q,P-p)$\\\midrule
1 & Application & 1 & 10 & $(13,12)$ & $(15,15)$ & $(2,3)$\\
2 & Application & 3 & 5 & $(6,9)$ & $(13,12)$ & $(7,3)$\\
3 & Renewal & 2 & 0 & $(2,3)$ & $(6,9)$ & $(4,6)$\\\bottomrule
\end{tabular}
\caption{Illustrative one-pair trace under \texttt{FULL-RESERVED}. Each disclosed value is $C_1[q]$ or $C_2[p]$; the integers themselves are recomputed by the verifier. The two-bit digest is pedagogical, not a security parameter recommendation.}\label{tab:trace}
\end{table}
The first record discloses $(C_1[13],C_2[12])$, verified by two and three evaluations to the registered heads. After its confirmed acceptance, the second discloses $(C_1[6],C_2[9])$, verified by seven and three evaluations to the first record's heads. Finally, the renewal record discloses $(C_1[2],C_2[3])$ under the old chains. If the replacement chains also have length 15, acceptance installs their epoch-one heads, resets $r$ to zero, and sets the global counter to three. The next ordinary record uses slot one of the new phase with global record number four.

Within the first slot, replacing $u=1$ by $u=2$ changes $(q,p)=(13,12)$ to $(12,13)$. Forward traversal can supply $C_2[13]$ from $C_2[12]$, but not the earlier indexed value $C_1[12]$ in the symbolic chain model. In contrast, both coordinates of the legitimate second-slot opening lie before those of the first. Figure~\ref{fig:mechanism} separates these two uses of the index geometry.

\begin{figure}[htbp]\centering
\includegraphics[width=0.96\textwidth]{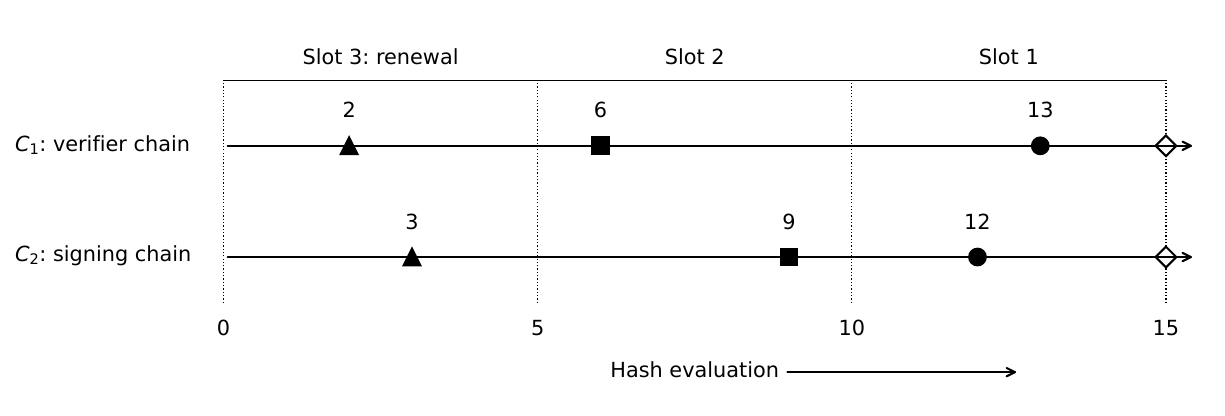}
\caption{The trace of Table~\ref{tab:trace}. Hash evaluation proceeds toward larger indices; disclosure proceeds through slots 1, 2, and 3 toward smaller indices. Circles mark the first record, squares the second, and triangles the old-chain openings of the renewal certificate. Open diamonds mark the registered heads. Segment boundaries are index coordinates, not published chain commitments.}\label{fig:mechanism}
\end{figure}
\section{Sequential operation and authenticated renewal}\label{sec:algorithms}\label{sec:renewal}\label{sec:policy}\label{sec:segment-state}
The following specification is for SEGMENTED. Its registered policy is either \texttt{FULL-RESERVED}, which reserves the last slot for renewal, or \texttt{TERMINAL}, which permits all slots to carry ordinary records and then stops. Early renewal is rejected under the former policy; the certificate occupies exactly slot $J$. This fixes the phase schedule used by the resource analysis.

Both profiles use the canonical signed record
\begin{equation}\label{eq:record}
R=(\tau,\mathsf{id},n+1,X(S),T,m,\boldsymbol\Omega),
\end{equation}
where $S$ is the authoritative predecessor and $T$ its history commitment. Indices are recomputed from the digest rather than transmitted. The openings are only one part of the record. After construction of $R$, define
\begin{equation}\label{eq:history}
T'=\hash\bigl(\Enc(\texttt{HIST},T,R)\bigr).
\end{equation}
This commitment is not included in its own input record. A new phase uses independently sampled seeds and incremented chain epochs; its local slot count resets, while the global record counter continues.

A renewal payload is $m=\Enc(W,0)$, with one entry for every chain, sorted by family and component. Each entry is
\begin{equation}\label{eq:replacement}
w=(a,i,e,e+1,H_{\rm old},L',H').
\end{equation}
Here $H_{\rm old}$ is the current head in $S$, and $H'=f_{a,i,e+1}^{L'}(s')$ is generated from a fresh secret seed. New lengths are positive, bounded by $L_{\max}$, equal within each pair, and provide $J'\ge2$. The nonce field is zero because every digest fits the designated slot. All $2k$ heads are replaced together. The certificate authorizes the advertised heads; it does not publicly prove knowledge of their seeds.

Registration and each subsequent operation follow these rules:
\begin{enumerate}
\item \emph{Register.} Choose the identity, profile, policy, block sizes, resource limits, and pair lengths. Require $J\ge2$ for \texttt{FULL-RESERVED} or $J\ge1$ for \texttt{TERMINAL}. Generate the $2k$ chains, authenticate $(G,T_0)$, and initialize the global counter $n=0$ and consumed-slot count $r=0$. For continuing operation, the payload limit must also accommodate the intended full replacement certificates. The signer durably stores its secret seeds or traversal data and a copy of the public state.
\item \emph{Prepare.} Obtain exclusive access to the signing state and confirm that it agrees with H. If a previously emitted record is pending, retransmit that exact record and emit no successor. Otherwise use slot $j=r+1$. Ordinary records are allowed for $j<J$ under \texttt{FULL-RESERVED}, and for $j\le J$ under \texttt{TERMINAL}. Slot $J$ under \texttt{FULL-RESERVED} requires the full renewal payload above. Reject incompatible types, exhausted phases, and oversized payloads before disclosure.
\item \emph{Open.} Compute $D(S,\tau,m)$, extract the registered blocks, and obtain $p_{j,i},q_{j,i}$ from~\eqref{eq:segment-encoding}. Set $\omega_i^{(1)}=C_{1,i,e}[q_{j,i}]$ and $\omega_i^{(2)}=C_{2,i,e}[p_{j,i}]$, and construct $R$. Run the checker below locally. Atomically persist $R$, its candidate successor, and all required secrets before publishing. The predecessor is now consumed for signing even while acceptance is pending.
\item \emph{Check.} The pure function $\mathsf{CheckStep}(S,R)$ rejects noncanonical encodings, incorrect identity, type, counter, or predecessor references, invalid slot/policy combinations, payloads exceeding $M_{\max}$, and opening vectors other than $k$ pairs of $\lambda$-bit strings. For renewal, it also checks exactly one entry per chain, ordering, old heads and epochs, incremented epochs, head sizes, length bounds, equal pair lengths, zero nonce, and $J'\ge2$. It recomputes the digest and indices and requires, for every component,
\[
f_{1,i,e}^{Q_i-q_{j,i}}(\omega_i^{(1)})=A_i,\qquad
f_{2,i,e}^{P_i-p_{j,i}}(\omega_i^{(2)})=B_i.
\]
On an ordinary record, it returns the state with the new indices and heads, $n+1$, $r+1$, and $T'$. On renewal, it first verifies these old-chain openings, then installs every new epoch, length, and head, sets the active indices to the new lengths, resets $r=0$, and retains $n+1$ and $T'$. Immutable registration parameters are unchanged. The checker itself makes no persistent write.
\item \emph{Accept and synchronize.} A verifier commits $(R,S')$ only if checking succeeds and its authoritative predecessor is still $S$. An identical accepted record is a duplicate, not another transition. H supplies ordered acceptance and confirmation. The signer confirms membership of the pending record, adopts the accepted successor as its synchronized state, and only then emits another record. A crash or missing confirmation permits exact retransmission, not rollback or a different record under the consumed predecessor.
\end{enumerate}
The public authoritative state and the signer's pending candidate are distinct until acceptance. Other verifiers synchronize through H rather than treating an unconfirmed publication as a new anchor. A joining auditor obtains the authenticated registration, the relevant record sequence, and an authenticated checkpoint, replays the transitions, and checks agreement. A trusted checkpoint that includes the current state can support current verification, while historical claims still require membership evidence. These interfaces specify what H must provide without choosing a consensus implementation.

\begin{proposition}[Correctness over a segmented phase]\label{prop:correctness}
For honestly generated chains and valid registered parameters, every policy-allowed payload within the size limit produces a record accepted by $\mathsf{CheckStep}$ on the same predecessor state. An ordinary record adopts the disclosed openings as heads. A last-slot renewal installs exactly the advertised full replacement and resets only the phase-local counter.
\end{proposition}
\begin{proof}
Theorem~\ref{thm:segment-capacity} assigns each digest a unique pair of interior positions below the current heads. Chain composition proves both traversal checks. Canonical inputs give the signer and checker identical indices. The ordinary update adopts those values. For renewal the same checks apply under the old domains before the validated replacement list is installed atomically. The global counter and history update follow the same record rule in both cases.
\end{proof}
With lengths 65535 and 32 pairs, the phase contains 254 ordinary records and one certificate. A terminal registration instead permits 255 ordinary records. Once a terminal phase is exhausted, no in-band renewal is possible. These are availability statements; certificate unforgeability is a separate security obligation.

\begin{proposition}[Conditional renewal continuity]\label{prop:renew-continuity}
Assume an authenticated registration, acceptance only of valid old-state renewal certificates, unique epoch progression, and atomic acceptance. Every accepted replacement head is connected to the registration by verifiable authorization records.
\end{proposition}
\begin{proof}
The registration is the base. A valid certificate binds each replacement to the previously linked state and its epoch. Atomic acceptance extends that sequence. Induction over replacements establishes the claim.
\end{proof}
The proposition concerns authorization continuity conditioned on validity. Section~\ref{sec:security} bounds impersonation only within a renewal-free horizon. The DIRECT profile uses the same record and acceptance interfaces but a different index rule and availability policy; its precise operational differences are given in Appendix~\ref{app:direct-operation}.
\section{Proportional allocation and its capacity limits}\label{sec:equivalence}\label{sec:recursion}\label{sec:capacity}\label{sec:duration}
DIRECT serves as the proportional comparison to the segmented construction. For fixed chains, coordinatewise order describes forward traversal of disclosed indices. The normalization in Theorem~\ref{thm:equivalence} gives the exact integer code; the following results describe what happens when accepted openings are reused as its endpoints.

Suppress the component and epoch subscripts. Given active indices $P,Q>0$, choose an interior signing index $p$ and set
\begin{equation}\label{eq:recurrence}
\boxed{q=Q\left(1-\frac pP\right),\qquad Qp+Pq=PQ.}
\end{equation}
Both indices must be integers. The candidate heads are $A'=C_1[q]$ and $B'=C_2[p]$. The transition uses prefixes of the same physical chains; it does not regenerate their seeds. In subsequent rounds, $(P,Q)$ is replaced by $(p,q)$.

\begin{proposition}[Within-epoch monotonicity]\label{prop:monotone}
If $0<p<P$ and $Q>0$, then $0<q<Q$. Consequently, both active indices decrease strictly along an interior renewal-free sequence.
\end{proposition}
\begin{proof}
The ratio $p/P$ lies in $(0,1)$, as does its complement. Multiplying the complement by positive $Q$ gives the asserted inequalities. Iteration proves the sequence statement.
\end{proof}

\begin{proposition}[Inductive verification]\label{prop:verification}
For honest chain values at valid integer indices,
\begin{equation}\label{eq:checks}
f_1^{Q-q}(A')=A,\qquad f_2^{P-p}(B')=B.
\end{equation}
Composing accepted transitions links every later head to the authenticated head at the beginning of its epoch.
\end{proposition}
\begin{proof}
The chain identity $f^{r-j}(C[j])=C[r]$ proves each equality. Over successive transitions, index differences telescope. Thus the composite traversal from a current index to the epoch's initial index equals the prescribed chain traversal.
\end{proof}

\begin{theorem}[Local incomparability]\label{thm:antichain}
For fixed $P,Q>0$, the set
\[
\mathcal A(P,Q)=\{(p,q)\in\mathbb R_{\ge0}^2:Qp+Pq=PQ\}
\]
is an antichain under coordinatewise order, as is any restriction to integer positions.
\end{theorem}
\begin{proof}
If $(p,q)\preceq(p',q')$ and both pairs satisfy the equation, subtraction gives $Q(p'-p)+P(q'-q)=0$. Both summands are nonnegative and both coefficients are positive. Hence both differences vanish.
\end{proof}
This theorem concerns alternatives under the same state. Successive legitimate states are comparable: both new indices are smaller. Taking the union of admissible families over different states does not produce a global antichain. Within a fixed epoch and without other leakage, the smallest disclosed index in each chain is its current active index. Earlier disclosed points add no earlier positions to its symbolic forward closure.

Applying Theorem~\ref{thm:equivalence} with $L=P$ gives the exact parametrization
\begin{equation}\label{eq:integer}
\boxed{(p,q)=\left(\frac Pd a,\frac Qd(d-a)\right),\qquad
 d=\gcd(P,Q),\quad a=0,\ldots,d.}
\end{equation}
There are $d-1$ strictly interior pairs. The endpoints disclose a seed of one chain and are excluded from the recursive signing profile. If $d=1$, no interior transition exists. If $P$ or $Q$ is zero, the state is outside the specified profile.

For a $b$-bit block $u$, the direct profile uses
\begin{equation}\label{eq:direct}
a=1+u,\qquad 0\le u<2^b.
\end{equation}
It is admissible exactly when $a<d$. Distinct admitted blocks yield distinct pairs. Universal block capacity therefore requires
\begin{equation}\label{eq:universal}
\boxed{d\ge 2^b+1.}
\end{equation}
For $b=16$, $d$ must be at least 65537. Equal lengths 65535 provide only 65534 interior pairs, not all 65536 blocks. Counting endpoints would change the seed-disclosure semantics rather than solve the reusable-profile problem.

Rounding the proportional line is not equivalent to this exact code. If $Q<P$ and every integer signing position from zero through $P$ is assigned an integer verifier position from zero through $Q$, some verifier index must repeat. Two pairs with equal verifier coordinate and unequal signing coordinate are comparable. Hence arbitrary rounding cannot preserve the same antichain over all those inputs.

For component $i$, let
\[
d_i=\gcd(P_i,Q_i),\qquad
c_i=\min\{2^b,\max\{0,d_i-1\}\}.
\]
With fixed block-to-component assignment, the number of admitted block vectors is
\begin{equation}\label{eq:capacity-product}
|\mathcal D_{\rm adm}(S)|=\prod_{i=1}^k c_i.
\end{equation}
This is a combinatorial count, not an independence claim about concrete hash outputs. If one block is inadmissible, the ordinary signing operation is declined without changing the state or releasing partial openings.

\begin{theorem}[Vector antichain]\label{thm:vector}
Fix a state and an injective block encoding into each component's local antichain. The resulting vectors for distinct admitted digests form an antichain in the coordinatewise order on all $2k$ indices.
\end{theorem}
\begin{proof}
If two vectors were comparable, the pair in each component would be comparable. Local incomparability forces equality in every component. Injectivity then forces equality of all blocks, contradicting distinct digests.
\end{proof}
Different messages with the same digest are not separated by this theorem. Message binding, collision resistance, and second-preimage resistance are separate computational questions.

\begin{theorem}[Product contraction]\label{thm:contraction}
Every interior transition under the proportional recurrence~\eqref{eq:recurrence} satisfies
\begin{equation}\label{eq:contraction}
pq\le\frac{PQ}{4}.
\end{equation}
After $t$ transitions without renewal, $P_tQ_t\le P_0Q_0/4^t$.
\end{theorem}
\begin{proof}
With $x=p/P\in(0,1)$, equation~\eqref{eq:recurrence} gives $pq=PQx(1-x)$. The inequality $x(1-x)\le1/4$ proves the first assertion; iteration proves the second.
\end{proof}
If universal block capacity is required after $t$ rounds, then $\gcd(P_t,Q_t)\ge2^b+1$ and consequently
\begin{equation}\label{eq:duration-bound}
P_0Q_0\ge4^t(2^b+1)^2.
\end{equation}
This is necessary, not sufficient: two large remaining indices may be coprime.

\begin{proposition}[A finite exact construction]\label{prop:finite}
Let $B=2^b+1$, $T\ge1$, and $P_0=uB^T$, $Q_0=vB^T$ for positive integers $u,v$. For any sequence $a_n\in\{1,\ldots,B-1\}$, set
\begin{equation}\label{eq:finite}
P_n=\frac{a_n}{B}P_{n-1},\qquad
Q_n=\frac{B-a_n}{B}Q_{n-1}.
\end{equation}
All $T$ transitions are integer, interior, proportional, and injective in the selected block at each fixed state.
\end{proposition}
\begin{proof}
For $n\le T$,
\[
P_n=uB^{T-n}\prod_{j=1}^n a_j,\qquad
Q_n=vB^{T-n}\prod_{j=1}^n(B-a_j).
\]
These are positive integers and strictly decrease at each step. The proportional equation follows from complementary factors. Distinct $a_n$ give distinct signing positions at a fixed state.
\end{proof}
This construction is not the direct profile in~\eqref{eq:direct}. In the exact parametrization its selected lattice parameter is $(d/B)a_n$, rather than $a_n$, whenever $B\mid d$. It proves existence of a finite recursive encoding, not efficiency of the operational direct profile. Already for $b=16,T=2$, $B^2=4\,295\,098\,369$ evaluations are required for a chain of that length. The construction does not justify practical indefinite reuse.

\begin{proposition}[Capacity after an equal-index direct transition]\label{prop:equal}
If $P=Q=L$, direct encoding gives $(p,q)=(a,L-a)$, and the next interior capacity is $\gcd(a,L)-1$. If $L$ is prime, any interior transition leaves zero capacity.
\end{proposition}
\begin{proof}
Here $d=L$, so~\eqref{eq:integer} gives $(a,L-a)$. Euclid's identity yields $\gcd(a,L-a)=\gcd(a,L)$, which is one for prime $L$ and $0<a<L$.
\end{proof}

\begin{corollary}[Loss of universal capacity]\label{cor:equal}
For direct $b$-bit encoding, a transition from equal active indices never leaves universal capacity, regardless of how large or composite $L$ is.
\end{corollary}
\begin{proof}
Since $a\le2^b$, one has $\gcd(p,q)=\gcd(a,L)\le2^b$, below the threshold in~\eqref{eq:universal}.
\end{proof}
Equal lengths 65537 admit all 16-bit blocks initially but no second interior direct signature. A full renewal as the first operation can replace these chains, but consumes their only operation without signing an ordinary message. Unequal indices can avoid this particular obstruction: with $b=2$, $(P,Q)=(77,44)$ and $a=4$, the result is $(28,28)$, which retains universal capacity. The next direct transition from that equal pair is again subject to the corollary. This example does not establish a generally efficient renewal policy.

\section{Conditional sequential security}\label{sec:history-assumption}\label{sec:security}
An auditor receives the initial authenticated registration, the record sequence, and an authenticated checkpoint $(n,T_n,X(S_n))$. It replays the transitions and requires agreement with the checkpoint. An internally consistent history with a different final commitment is not evidence of membership in the authoritative history. An offline verifier that has only an initial registration may lack sufficient information to choose between retrospective alternatives.

\begin{proposition}[Retrospective forward reconstruction]\label{prop:retrospective}
After openings at indices $(P_t,Q_t)$ have been disclosed within a fixed epoch, all pairs of chain values at indices $(p,q)$ with $p\ge P_t$ and $q\ge Q_t$ can be reconstructed by forward evaluation, even if they represent alternatives under an earlier state.
\end{proposition}
\begin{proof}
Apply $p-P_t$ evaluations in the signing chain and $q-Q_t$ evaluations in the verifier chain. Both exponents are nonnegative.
\end{proof}
Later openings may therefore reconstruct alternatives to historical records. Under H, such reconstructed values do not establish membership in the accepted prefix. This distinction motivates authenticated history without assuming the legitimacy of a newly submitted extension.

Fix a finite horizon $N$ without chain replacement. Impersonation means acceptance of a current-state extension not emitted by the honest signer. An identical retransmission is not a new forgery. A round is the interval during which one authoritative predecessor state is eligible for extension. An adversary observes the history, retains computational state allowed by the model, adaptively requests messages, and schedules candidate submissions. At most one honest record is emitted under that predecessor, and its disclosure is visible even before acceptance. Exact retransmission of this emitted record is excluded from the forgery event; a distinct accepted record is not. The signer emits no successor-state record until H confirms the pending acceptance. Each round's attack budget includes failed verification attempts and all queries, including those before any honest emission. Declined signing requests reveal no openings and consume no state, but their responses and the adversary's computations remain part of its view. No liveness guarantee is made against an adversary that delays delivery indefinitely.

The renewal-free analysis quantifies over finite admissible horizons for the chosen parameters. It does not assert indefinite postponement at fixed parameters: designing a new, longer chain for each desired horizon is a different quantifier order from obtaining unbounded service from one fixed chain. Exhaustion is a failure of availability rather than successful impersonation. In addition to H, the composition theorem assumes:
\begin{enumerate}[label=\textbf{P\arabic*.}]
\item \textbf{Authentic initialization and exclusive use.} The initial registration is authentic, seeds remain secret, and each state is consumed at most once to emit a record. There is no rollback or emission of incompatible records under the same predecessor.
\item \textbf{Invariant transition rules.} Every round uses the registered profile and encoding, binds the message to its state, and applies the prescribed opening checks. Honest transitions use integer interior positions and update all components atomically.
\item \textbf{Uniform one-step security, B.} For every allowed adversary and round $j$, conditional on no previous accepted impersonation, the probability of a first unauthorized accepted extension in that round is at most $\varepsilon_j\in[0,1]$. This guarantee covers the actual reachable-state distribution and the adversary's accumulated view, including chosen messages, prior queries, computations, decline responses, and any honest record emitted under the current predecessor before acceptance. The attack budget within the round is included in $\varepsilon_j$.
\item \textbf{Admissible horizon.} The $N$ honest transitions under consideration have sufficient capacity and require no renewal. A decline changes neither the public state nor the adversary's disclosed opening set.
\end{enumerate}
B requires a bound for exposed prefixes under the actual adaptive transcript, not only for newly sampled initial chains. The theorem composes those conditional bounds; deriving them from a hash-function assumption is a separate reduction problem. This distinction separates the deterministic disclosure invariant from its computational security interpretation.

For example, fix a state and type and suppose an adversary finds distinct admissible payloads $m,m'$ with $D(S,\tau,m)=D(S,\tau,m')$. After obtaining the honest record for $m$, it can reuse its entire opening vector in a record for $m'$ under the same predecessor. All recomputed indices and chain checks coincide. If the substituted record is accepted first, H preserves that acceptance and does not identify the substitution. This conditional attack explains why digest binding must be covered by B: incomparability excludes distinct coordinatewise forward substitutions, but does not distinguish equal digests. It is not an assertion that such a collision can be found efficiently for the concrete profile.

\begin{lemma}[Prefix and disclosure invariant]\label{lem:prefix}
After $n$ honest renewal-free transitions, the active heads in every component are $C_{1,i}[Q_{n,i}]$ and $C_{2,i}[P_{n,i}]$. Every earlier disclosed point in those physical chains has index at least the respective current index and is forward-derivable from that head. The next transition uses the registered index rule with $S_n$ as its reference state: the proportional equation for DIRECT or the next translated codebook for SEGMENTED. The opening checks retain their form.
\end{lemma}
\begin{proof}
At $n=0$ the heads are those of the authenticated registration. Assume the statement at $n$. For DIRECT, Proposition~\ref{prop:monotone} gives strict decrease; for SEGMENTED, Theorem~\ref{thm:segment-capacity} gives it. Thus
\[
0<P_{n+1,i}<P_{n,i},\qquad 0<Q_{n+1,i}<Q_{n,i}.
\]
The update adopts the openings at these indices as the new heads. For any historical $r\ge Q_{n,i}$,
\[
C_{1,i}[r]=f_{1,i}^{\,r-Q_{n+1,i}}(C_{1,i}[Q_{n+1,i}]),
\]
and the signing chain has the analogous identity. Finally, the registered rule determines the next admissible indices: equation~\eqref{eq:recurrence} for DIRECT or equation~\eqref{eq:segment-encoding} for SEGMENTED. The two chain-verification equations retain their form. This proves the invariant by induction for all components.
\end{proof}
In the symbolic closure model, previous openings therefore add no earlier point than the current heads. The lemma explains why the directional barrier is preserved. It does not assert that seeds conditioned on the full transcript have the same distribution as newly generated seeds; the computational consequences of auxiliary information are covered by B.

\begin{theorem}[Conditional composition across rounds]\label{thm:security}
Under H and P1--P4, let $F_N$ denote at least one accepted impersonation within the first $N$ rounds. Then
\begin{equation}\label{eq:security}
\boxed{\Pr[F_N]\le1-\prod_{j=1}^{N}(1-\varepsilon_j)
\le\min\left\{1,\sum_{j=1}^{N}\varepsilon_j\right\}.}
\end{equation}
If $\varepsilon_j\le\varepsilon$ uniformly, then $\Pr[F_N]\le\min\{1,N\varepsilon\}$. Independence between rounds is not required.
\end{theorem}
\begin{proof}
\emph{Base case.} P1 authenticates $S_0$ and B gives $\Pr[F_1]\le\varepsilon_1$.

\emph{Induction hypothesis.} Suppose $\Pr[F_n]\le1-\prod_{j=1}^{n}(1-\varepsilon_j)$. Conditional on $\overline F_n$, every accepted record so far was emitted by the signer, and H fixes the authoritative order and current state. Lemma~\ref{lem:prefix} preserves the prefix structure; P1 and P2 prevent incompatible reuse or changes to the verification rule.

\emph{Inductive step.} Let $E_{n+1}$ be a first unauthorized acceptance in the next round. When $\Pr[\overline F_n]>0$, B yields $\Pr[E_{n+1}\mid\overline F_n]\le\varepsilon_{n+1}$. Consequently,
\begin{align*}
\Pr[F_{n+1}]&=\Pr[F_n]+\Pr[\overline F_n]\Pr[E_{n+1}\mid\overline F_n]\\
&\le\Pr[F_n]+(1-\Pr[F_n])\varepsilon_{n+1}\\
&\le1-\prod_{j=1}^{n+1}(1-\varepsilon_j).
\end{align*}
The last inequality uses monotonicity of $x+(1-x)\varepsilon_{n+1}$ in $x$. If $\Pr[\overline F_n]=0$, the corresponding bound is trivial under the induction hypothesis. Finally, $1-\prod_j(1-\varepsilon_j)\le\sum_j\varepsilon_j$ follows by induction, and probabilities are bounded by one.
\end{proof}

\begin{corollary}[Polynomial horizon]
For a security-parameter family indexed by $\lambda$, if $N(\lambda)$ is polynomial and one negligible function $\varepsilon(\lambda)$ uniformly bounds all rounds, the impersonation advantage over that horizon is negligible.
\end{corollary}
\begin{proof}
A polynomial times a negligible function is negligible. Apply the uniform bound in Theorem~\ref{thm:security}.
\end{proof}
The concrete 256-bit profile is not assigned a numerical security level by this asymptotic statement. A small bound for infinitely many attempts with positive per-round success probability also does not follow. The theorem removes the need for a new composition proof for each round count, but not the need to justify its one-step premise.

H is ideal in the preceding theorem. If an implementation agrees with the ideal service until an event $B_H$ of probability at most $\delta_H(N)$, an identical-until-failure argument and a union bound give
\[
\Pr[F_N]\le\min\left\{1,\delta_H(N)+\sum_{j=1}^{N}\varepsilon_j\right\}.
\]
Authentic initialization and exclusive state use remain assumptions; modeled failures of those mechanisms require additional bounds. A hash-linked list alone is not assigned a negligible $\delta_H$.

HSS also authenticates successor public keys, but through its specified tree hierarchy~\cite{RFC8554}; the analogy does not prove security of a RecHash replacement certificate. The theorem does not include \texttt{RENEW} transactions. Extending it to them requires a one-step guarantee for renewal certificates and the post-replacement state distribution. DIRECT additionally requires analysis of adaptive nonce selection; SEGMENTED has no admissibility search. A quantum-adversary version would likewise require B and H to hold against that adversary class; the probabilistic composition argument alone does not establish those premises. In a random-oracle treatment, quantum-accessible queries must be included explicitly~\cite{Boneh2011}; multi-target effects also require concrete analysis~\cite{Perlner2022}.

For SEGMENTED the general theorem specializes to the policy-limited horizon below. Later openings can reconstruct values in earlier segments, so authenticated history remains necessary for membership claims.

\begin{corollary}[Segmented renewal-free composition]\label{cor:segment-security}
Assume H, authentic initialization, exclusive state use, the registered segmented transition rules, and a uniform one-step security premise $\mathrm B_{\rm seg}$ for the actual accumulated transcripts. For an admissible ordinary-record horizon, take $N\le J-1$ under \texttt{FULL-RESERVED}, or $N\le J$ under \texttt{TERMINAL}. The reserved renewal certificate is excluded. The impersonation probability is at most
\[
1-\prod_{j=1}^{N}(1-\varepsilon^{\rm seg}_j)
\le\min\{1,\textstyle\sum_{j=1}^{N}\varepsilon^{\rm seg}_j\}.
\]
\end{corollary}
\begin{proof}
Use Theorem~\ref{thm:segment-capacity} for the prefix invariant and $\mathrm B_{\rm seg}$ for the conditional first-failure bound. The base case and probability recurrence in Theorem~\ref{thm:security} then apply verbatim; neither uses product contraction.
\end{proof}
A one-step result for DIRECT does not automatically establish $\mathrm B_{\rm seg}$. Similarly, the corollary does not cover renewal certificates or claim a numerical security level. Short digest blocks are position encodings, not truncated chain outputs; security must be assessed for the full digest, all openings, and the entire transcript. Increasing $k$ also increases the number of exposed hash values and potential attack targets, which a concrete reduction must account for.

\begin{proposition}[Finite renewal-free horizon]\label{prop:finite-horizon}
In either profile, a renewal-free sequence of $N$ honest interior transitions from initial active indices $(P_{0,i},Q_{0,i})$ satisfies
\[
N\le\min_i\{P_{0,i}-1,Q_{0,i}-1\}.
\]
For SEGMENTED, the sharper exact slot bound is $J$; the ordinary-record bound is $J-1$ when a renewal slot is reserved.
\end{proposition}
\begin{proof}
Each positive integer index decreases by at least one at every transition and remains at least one. Thus $1\le P_{N,i}\le P_{0,i}-N$, and similarly for $Q$. The segment bound and the reservation follow from Theorem~\ref{thm:segment-capacity} and the registered policy.
\end{proof}
Consequently, availability over a longer operating lifetime requires fresh authenticated material or a separately specified mechanism. Even if a one-step bound were uniform, its accumulation over a growing horizon would still have to be controlled; a positive fixed error bound does not yield a horizon-independent small failure probability.

\section{Parameter selection and resource analysis}\label{sec:complexity}\label{sec:segment-cost}
Table~\ref{tab:segments} keeps $\ell=\lambda=256$ and $L=65535$ while changing the component count. These are calculations for the segmented extension, not performance measurements reported in 2019.
\begin{table}[htbp]\centering\small
\begin{tabular}{@{}rrrrrr@{}}\toprule
Pairs $k$ & Physical chains & Bits/block & Steps/segment & Slots $J$ & Openings (KiB)\\\midrule
16 & 32 & 16 & 65537 & 0 & 1\\
32 & 64 & 8 & 257 & 255 & 2\\
64 & 128 & 4 & 17 & 3855 & 4\\
128 & 256 & 2 & 5 & 13107 & 8\\
256 & 512 & 1 & 3 & 21845 & 16\\\bottomrule
\end{tabular}
\caption{Capacity and opening payload at fixed physical-chain length. One KiB is 1024 bytes; metadata are excluded. The zero in the first row means that this strictly interior encoding needs two additional steps, not that the 2019 protocol has zero capacity.}\label{tab:segments}
\end{table}
For $k=32$, $65535=255\cdot257$. Slot 1 has boundaries $(65278,65535)$ and interior indices $65279,\ldots,65534$; slot 2 has boundaries $(65021,65278)$; slot 255 has boundaries $(0,257)$ and interior indices $1,\ldots,256$. The seed is never opened. Every slot accommodates all 256 block values without rounding or nonce search. For example, in slot 1, $u=0$ gives $(p,q)=(65279,65534)$ and $u=255$ reverses these coordinates.

Exactly 256 slots require $L=256\cdot257=65792$, an increase of 257 evaluations per chain, approximately $0.392\%$. Unequal block sizes also provide intermediate choices: with $k=48$, 16 blocks have six bits and 32 have five bits, giving $\lfloor65535/65\rfloor=1008$ slots at a three-KiB opening payload. With tailored lengths for those same 1008 slots, the six-bit pairs need 65520 steps and the five-bit pairs need 33264 steps. This avoids unused short-block capacity inherent in a common length.

For $\lambda$-bit chain outputs, the opening payload is $2k\lambda$ bits per record. The head portion of the current public state has the same size. Initial construction requires exactly $G=2\sum_i L_i$ chain-function evaluations, excluding seed generation and commitments. A full replacement has the same cost for the new lengths. These counts are not wall-clock timings.

For $k=32$, $\lambda=256$, and $L_i=65535$, $G=4194240$ evaluations and the opening payload is 2048 bytes. Storing all positive-index chain values uses $2kL\lambda/8=134215680$ bytes, approximately 128 MiB, in addition to seeds and metadata. Keeping only the 64 independent seeds uses 2048 bytes for those seeds but requires recomputation. Caching or chain traversal algorithms provide other tradeoffs that are not benchmarked here. An archive of $N$ opening vectors alone takes $N\,2k\lambda$ bits.

\begin{proposition}[Message-independent chain-evaluation counts]\label{prop:segment-cost}
In a segmented phase, verification of the first record uses
\[
V_1=\sum_i s_i
\]
chain evaluations; each subsequent record uses $V_j=2\sum_i s_i$. After $t\ge1$ records, the cumulative count is $(2t-1)\sum_i s_i$. If signing independently recomputes each opening from its seed, its count in slot $j$ is
\[
A_j=\sum_i\bigl[2L_i-(2j-1)s_i\bigr].
\]
These expressions exclude digest evaluation, record processing, and any additional local verification before publication.
\end{proposition}
\begin{proof}
At initialization, $P_{0,i}+Q_{0,i}=2L_i$. Equation~\eqref{eq:translated-sum} gives $p_{j,i}+q_{j,i}=2L_i-(2j-1)s_i$. Subtracting the new sum from the preceding sum yields $s_i$ in the first slot and $2s_i$ in every later slot. Seed-only signing is the sum of the two opening indices. 
\end{proof}
For the 32-pair configuration, the first verification takes 8224 chain evaluations and each later verification takes 16448. Across 255 records, verification takes $509\cdot8224=4186016$ evaluations. Seed-only signing takes 4186016 evaluations for the first record and 8224 for the last. Its total over $t$ records is $\sum_i(2tL_i-t^2s_i)$. Thus small verification counts do not imply equally inexpensive seed-only signing.

With one reserved certificate and full-length replacement, the generation cost charged per application record is
\begin{equation}\label{eq:amortization}
\frac{2\sum_i L_i}{J-1},\qquad J\ge2.
\end{equation}
This accounts for one chain batch per phase in repeated operation; startup and the final partially used phase must be accounted for separately in a finite deployment. For 32 pairs and 254 application records, it is $4194240/254\approx16512.76$ evaluations per application record. Renewal adds one opening vector per phase, so the amortized opening traffic, excluding metadata and replacement-head payloads, is $(2k\lambda/8)J/(J-1)$ bytes per application record. Each certificate also carries all replacement heads, lengths, epochs, and identifiers; those fields must not be hidden inside the opening-only size.

A useful parameter choice therefore depends on at least four constraints: the desired application-record horizon, the permitted opening payload, initialization work, and traversal memory. If the opening budget is $B$ bytes and $\lambda$ is divisible by eight, then $k\le\lfloor4B/\lambda\rfloor$. For $R$ application records with in-band renewal, replace $R$ by $R+1$ in~\eqref{eq:target-length}. For example, a 2048-byte opening budget permits at most 32 pairs with 256-bit outputs, and 256 application records plus renewal require length 66049 in the uniform 8-bit allocation.

At fixed physical length, more pairs increase initialization work linearly but can increase the slot horizon sharply. At a fixed target horizon, shorter blocks can instead reduce total generation work, as established by the balanced-allocation analysis. Neither comparison alone defines an optimal deployment. Longer admissible lifetimes also enlarge the transcript that the one-step security premise must cover. The calculations provide explicit operating points and constraints; measured timings and a concrete cryptographic reduction remain separate requirements.

For scale, FIPS 205 lists complete signature lengths of 7856 bytes for SLH-DSA-128s and 29792 bytes for SLH-DSA-256s, with 32-byte and 64-byte public keys, respectively~\cite[Table 2]{FIPS205}. RecHash's 2048-byte figure counts openings only: it excludes metadata, registration, renewal payloads, and authenticated-history infrastructure. Its current head vector also contains 2048 bytes before metadata. These figures compare different verification interfaces and have no established equal-security interpretation; they do not demonstrate a signature-size advantage over SLH-DSA.

The lifetime constraint is likewise construction-specific. The draft SP 800-230 explores a bounded-use SLH-DSA regime~\cite{NIST230}; its cap is not a physical segment count and must not be substituted into equation~\eqref{eq:balanced}. For a deployment comparison, the same application workload must be charged for signing, verification, synchronization, storage, and key replacement. A checkpoint authenticated by an external signature scheme would add that scheme's costs rather than eliminate the history assumption.

Figures~\ref{fig:lifetime} and~\ref{fig:tradeoff} show the fixed-length and fixed-horizon tradeoffs with 256-bit outputs and two physical chains per component. Fifteen reproducible checks cover the integer formulas, segmented allocations, toy chain traversals, the worked renewal trace, probability recurrence, and plotted datasets; they do not establish the one-step hardness premise or validate a production implementation.

\begin{figure}[!htbp]
\centering\includegraphics[width=0.76\textwidth]{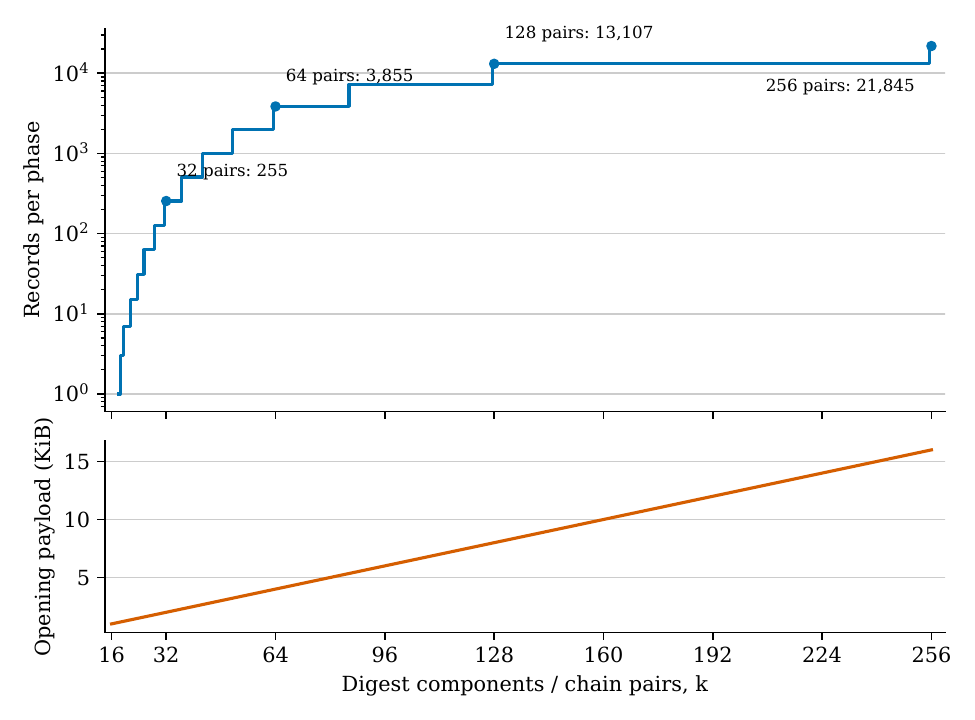}
\caption{Fixed-length parameter analysis for $\ell=\lambda=256$ and $L=65535$. Upper panel: scheduled records; renewable phases reserve one of these records for their certificate. Zero-capacity cases are omitted from the logarithmic ordinate. Lower panel: opening payload, excluding metadata. Curves evaluate the stated formulas for every integer $16\le k\le256$; labeled points identify selected configurations.}\label{fig:lifetime}
\vspace{1em}
\centering\includegraphics[width=0.76\textwidth]{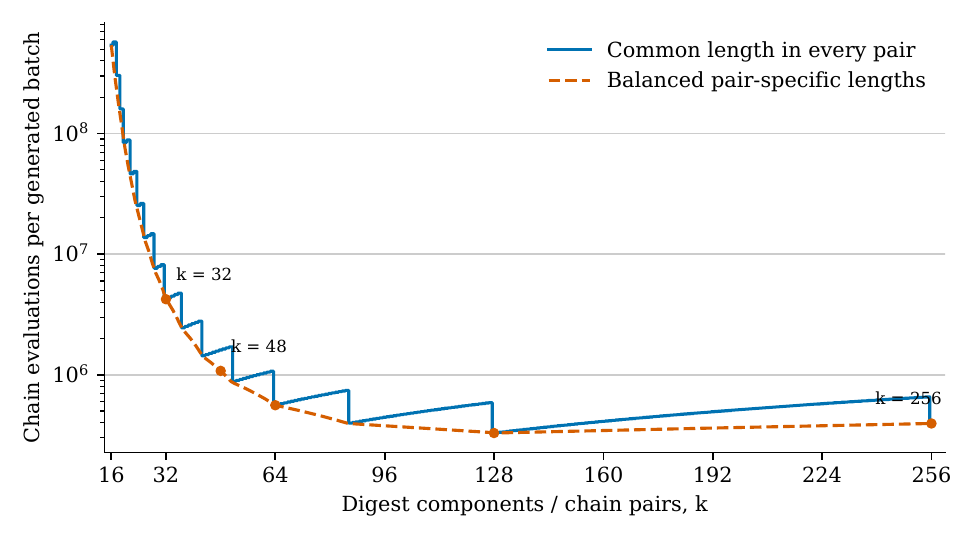}
\caption{Generation cost for a fixed horizon of 256 application records and one renewal certificate. The common-length design uses $L=257(2^{\lceil256/k\rceil}+1)$ in every chain. The pair-specific design uses balanced block sizes and $L_i=257(2^{b_i}+1)$. Each curve counts the evaluations needed to construct all $2k$ chains once. These are analytic counts, not timing measurements; opening payload remains $64k$ bytes.}\label{fig:tradeoff}
\end{figure}
\clearpage
\section{Conclusion}
The decreasing position relation recovered from the 2019 hash-entangled construction becomes an exact two-coordinate constant-sum antichain under integer normalization. This local equivalence connects the historical design to the broader study of incomparable chain encodings while preserving the distinction between complete protocols. DIRECT demonstrates the arithmetic and contraction constraints of proportional reuse. SEGMENTED retains local incomparability and places successive codes in disjoint intervals, giving a deterministic phase capacity and an explicit renewal reserve.

Balanced partitioning relates that capacity to the number of physical chains, transmitted openings, and generation work. At length 65535, 32 pairs support 254 application records and a renewal certificate with a two-KiB opening vector per record. The protocol requires authenticated order, serialized disclosure, and synchronized acceptance. Its probability-composition result is conditional on an adaptive one-step guarantee and applies within a renewal-free horizon. Concrete reductions, security across replacement, and measured implementation costs are the remaining steps toward assessing post-quantum deployment.
\appendix
\section{Canonical serialization}\label{app:encoding}

$\Enc$ encodes trees with byte-string, nonnegative-integer, and list types. Type tags are the single bytes 0, 1, and 2. A byte string consists of tag 0, its length in eight big-endian bytes, and its content. An integer consists of tag 1, the length of its minimal big-endian magnitude in eight bytes, and that magnitude; zero uses one zero byte. A list consists of tag 2, its element count in eight big-endian bytes, and the concatenated encodings of its elements. Objects must fit the length fields and registered resource limits. Leading zero bytes in nonzero integer magnitudes, inconsistent list counts, and trailing data are rejected.

Labels are ASCII byte strings. $\Enc(x_1,\ldots,x_r)$ denotes the encoding of a list. Applications convert text to bytes before signing; no implicit text normalization occurs after signing. The parameter object is
\[
(\texttt{RecHash-v2},k,b,\lambda,L_{\max},M_{\max},N_{\max},\texttt{DIRECT}).
\]
The DIRECT profile retains this protocol identifier and encoding. For SEGMENTED, replace the parameter object by $(\texttt{RecHash-seg-v1},\ell,\lambda,(b_i)_{i=1}^k,L_{\max},M_{\max},\pi)$, where $\pi$ is the ASCII policy label \texttt{FULL-RESERVED} or \texttt{TERMINAL}, $k$ is the list length, $\sum_i b_i=\ell$, and $\ell=\lambda=256$ in the concrete profile. Validate $1\le k\le\ell$, positive block sizes, and registered size limits. Append the epoch-local slot counter $r$ to the ordinary state list; the complete object is covered by its state digest. A SEGMENTED registration includes this profile object and initializes the state with $r=0$. The list of physical lengths determines $J$. The canonical record and history-update formats are unchanged. Profile separation is authenticated through the state; switching profiles requires a separately authenticated registration. Write each component as the list $((e_1,L_1,Q,A),(e_2,L_2,P,B))$, and order the component list by increasing $i$. The registration object is $G=(\mathsf{params},\mathsf{id},\mathsf{components})$, with zero epochs and active indices equal to initial lengths. State serialization is $(\mathsf{params},\mathsf{id},n,T,\mathsf{components})$ for DIRECT, with $r$ appended for SEGMENTED. Registration initializes $n=r=0$, with $r$ omitted for DIRECT.  A renewal payload encodes $(W,\nu)$ with replacement entries in the field order of~\eqref{eq:replacement}. The signed record follows~\eqref{eq:record}. The verifier requires canonical encoding of embedded renewal data as well as of the outer record.

The default segmented policy is \texttt{FULL-RESERVED}. A separately registered \texttt{TERMINAL} policy permits $J$ ordinary records and no renewal. Policy changes within a registered phase are not allowed. The transition rules of Section~\ref{sec:algorithms} apply with the registered profile; its checker rejects stale predecessors, skipped slots, profile mismatches, and unsupported replacement patterns. State must be committed durably before publication, and an identical retransmission must not consume another slot. Stateful-signature guidance provides implementation context for exclusive state use~\cite{NIST208,RFC8391}, without transferring its security guarantees to RecHash.
\section{DIRECT operation and renewal policy}\label{app:direct-operation}
DIRECT retains the record format~\eqref{eq:record}, history update~\eqref{eq:history}, exclusive state use, and acceptance synchronization of Section~\ref{sec:algorithms}. It uses uniform $b$-bit blocks with $kb=\ell$. For each component, the checker computes $d_i=\gcd(P_i,Q_i)$ and $a_i=1+u_i$, rejects $a_i\ge d_i$, and sets $p_i=(P_i/d_i)a_i$ and $q_i=(Q_i/d_i)(d_i-a_i)$. It verifies the same two traversals and adopts the accepted indices and openings. There is no segment counter or slot schedule.

Define $\mathsf{Full}(S)$ by $\gcd(P_i,Q_i)\ge2^b+1$ for every component. The reference signer emits an ordinary record only if its digest is admissible and its candidate successor satisfies $\mathsf{Full}$. Otherwise it returns \textsf{RENEW\_FIRST} without disclosure. For renewal it tests $\mathsf{Full}$ after applying all replacements. These signer-side conditions preserve universal successor capacity when satisfied; they do not guarantee progress. In particular, Corollary~\ref{cor:equal} rejects every ordinary transition from a state containing an equal-index pair under this policy.

A DIRECT renewal payload is $\Enc(W,\nu)$ with a nonempty replacement list in the field order~\eqref{eq:replacement}, sorted by $(a,i)$ and without duplicates, and $0\le\nu<N_{\max}$. The checker validates identifiers, current heads and epochs, incremented epochs, $\lambda$-bit replacement heads, and positive lengths bounded by $L_{\max}$. The signer generates independent replacement seeds and searches distinct nonces under the fixed old state and fixed list. A candidate is admissible only if all its blocks satisfy $1+u_i(\nu)<d_i$. A failed bounded search publishes nothing and replaces nothing.

All entries are authorized in one record under the old state. After the opening checks, the checker first forms the consumed successor, then simultaneously installs the listed replacements at full length. Unreplaced chains remain at their consumed indices. A replacement head need not be an image of the old head. The signer persists the record, successor, and required secrets before publication and waits for confirmation exactly as in the segmented protocol. Partial replacement may fail to restore capacity, and $d_i=1$ cannot be repaired by nonce search.

\begin{proposition}[Ideal admissibility probability]\label{prop:nonce}
Fix $S$ and $W$. If distinct, previously unqueried trial inputs have independent uniform $kb$-bit digests, with $c_i=\min\{2^b,\max\{0,d_i-1\}\}$, then
\begin{equation}\label{eq:nonce-prob}
\rho(S)=\prod_{i=1}^k\frac{c_i}{2^b},\qquad
\Pr[\text{an admissible candidate in }M\text{ trials}]=1-(1-\rho(S))^M.
\end{equation}
For $\rho>0$, an unbounded ideal search has mean $1/\rho$ trials.
\end{proposition}
\begin{proof}
Independent blocks are admissible with probabilities $c_i/2^b$, giving the product. Independent trial failures give $(1-\rho)^M$, and the waiting time is geometric.
\end{proof}
This is an ideal search calculation, not a concrete SHA-256 guarantee. If replacement lengths make $\mathsf{Full}$ independent of the nonce, admissibility suffices for the final capacity test. Otherwise $\rho$ is only an upper bound on acceptance under the complete signer policy. Neither this probability nor renewal continuity proves security of adaptive nonce search.

\end{document}